\documentclass[conference]{IEEEtran}
\usepackage{cite} 
\usepackage{units}

\usepackage{amsmath} 
\usepackage{amssymb,amsthm} 
\usepackage[subnum]{cases}
\usepackage{ctable}
\usepackage{graphicx,graphics,color}
\graphicspath{{figures/}} 

\usepackage[caption=false,font=footnotesize]{subfig}

\newcommand{\Pul}{P_{\text{UL}}}
\newcommand{\Pc}{P_{\text{C}}}
\newcommand{\hdl}{h_{\text{DL}}}
\newcommand{\hul}{h_{\text{UL}}}
\newcommand{\Omegaul}{\Omega_{\text{UL}}}
\newcommand{\Omegadl}{\Omega_{\text{DL}}}
\newcommand{\Pdl}{P_{\text{DL}}}

\newcommand{\dd}{\text{d}}
\newcommand{\E}{\mathbb{E}}
\newcommand{\e}{{\rm e}}
\newcommand{\req}{\overset{\underset{!}{}}{=}}

\theoremstyle{plain}
\newtheorem{corollary}{Corollary}
\newtheorem{theorem}{Theorem}

\newtheorem{proposition}{Proposition}

\theoremstyle{definition}

\theoremstyle{remark}
\newtheorem{remark}{Remark}

\begin{document}
\title{Performance Analysis of Wireless Powered Communication with Finite/Infinite Energy Storage} 

\author{\IEEEauthorblockN{Rania Morsi, Diomidis S. Michalopoulos, and Robert Schober}
\IEEEauthorblockA{Institute of Digital Communications, Friedrich-Alexander-University Erlangen-N\"urnberg (FAU), Germany}
}


\maketitle

\begin{abstract}
In this paper, we consider an energy harvesting (EH) node which harvests energy from a radio frequency (RF) signal broadcasted by an access point (AP) in the downlink (DL). The node stores the harvested energy in an energy buffer and uses the stored energy to transmit data to the AP in the uplink (UL). We consider a simple transmission policy, which accounts for the fact that in practice the EH node may not have knowledge of the EH profile nor of the UL channel state information. In particular, in each time slot, the EH node transmits with either a constant desired power or a lower power if not enough energy is available in its energy buffer. For this simple policy, we use the theory of discrete-time continuous-state Markov chains to analyze the limiting distribution of the stored energy for finite- and infinite-size energy buffers. Moreover, we take into account imperfections of the energy buffer and the circuit power consumption of the EH node. For a Rayleigh fading DL channel, we provide the limiting distribution of the energy buffer content in closed form. In addition, we analyze the average error rate and the outage probability of a Rayleigh faded UL channel and show that the diversity order is not affected by the finite capacity of the energy buffer.  Our results reveal that the optimal desired transmit power by the EH node is always less than the average harvested power and increases with the capacity of the energy buffer.
\end{abstract}

\section{Introduction}
The performance of battery-powered wireless communication networks, such as sensor networks, is limited by the lifetime of the network nodes. Periodic replacement of the nodes' batteries is costly, inconvenient, and sometimes impossible when the sensor nodes are placed in a hazardous environment or embedded inside the human body. The lifetime bottleneck problem of energy-constrained wireless networks thus demands harvesting energy from renewable energy sources (e.g., solar, wind, thermal, vibration) to ensure a sustainable network operation. The harvested energy can then be used by the energy harvesting (EH) node to transmit data to its designated receiver. However, the aforementioned energy sources are in general intermittent and uncontrollable. For example, solar and wind energy are weather dependent and are not available indoors. In contrast, radio frequency (RF) energy is a viable energy source which is partially controllable and can be provided on demand to charge low-power devices \cite{Kansal_2007}. 

A common feature of EH communication networks is the randomness of the amount of harvested energy. 
For instance, solar/wind energy varies throughout the day and the harvested energy from an RF signal varies due to time-varying fading. Furthermore, the information signal transmitted by the EH node encounters also time-varying fading, which introduces another source of randomness. Therefore, one main objective of energy management polices for EH networks is to match the energy consumption profile of the EH node to the random energy generation profile of the EH source and to the random information channel \cite{Kansal_2007,WPC_TDMA,WPC_SDMA,Sharma2010,Ulukus_Yener_2011}. For example, the authors of \cite{Kansal_2007} introduced the concept of energy neutral operation of an EH system, where the energy used by the system is always less than the energy harvested. Energy neutrality is thus a condition for an EH system to operate perpetually. 
In \cite{WPC_TDMA,WPC_SDMA}, a harvest-then-transmit protocol is considered for a multiuser system with RF wireless power transfer (WPT) in the downlink (DL) and wireless information transfer (WIT) in the uplink (UL), where the users' sum rate or equal throughput is maximized on a per-slot basis. In \cite{Sharma2010}, throughput and mean delay optimal energy neutral policies, which stabilize the data queue of an EH sensor node over an infinite horizon, are proposed  in a time-slotted setting. In \cite{Ulukus_Yener_2011}, optimal transmission policies that maximize the throughput by a deadline or minimize the transmission completion time are proposed for an EH node with finite energy storage in a continuous time setting. 

Optimal offline transmission policies typically require non-causal knowledge of energy and channel state information (CSI) at the EH node, whereas optimal online solutions are typically based on dynamic programming which is computationally intensive even for a small number of transmitted symbols, see  \cite{Ulukus_Yener_2011} and the references therein. Therefore, these optimal policies may not be feasible in practice. For example, typical EH wireless sensor networks are expected to comprise many small, inexpensive sensors with limited computational power and energy storage. In such networks, even causal CSI may not be available at the EH nodes nor at the EH source.

Motivated by these practical considerations, in this paper, we consider a simple online transmission policy, where the CSI and the EH profile are not available at the EH node nor at the EH source. In particular, an access point (AP) transmits an RF signal with a constant power in the DL and the EH node harvests the received RF energy and uses the stored energy to transmit data to the AP in the UL. In each time slot, the EH node transmits with either a constant desired power or a lower power if not enough energy is available in its energy buffer. We model the stored energy by a discrete-time continuous-state Markov chain and provide its limiting distribution for both infinite and finite energy storage, when the DL channel is Rayleigh fading. Under this framework, we analyze the average error rate (AER) and the outage probability of a Rayleigh fading information link. We show that, surprisingly, the diversity order is not affected if the energy storage has finite capacity. Furthermore, we show that the optimal desired UL power of the considered policy is always less than the average harvested power and increases with the capacity of the energy buffer. The proposed framework also takes into account the system non-idealities such as non-zero circuit power consumption and imperfections of the energy buffer. 

The rest of the paper is organized as follows. Section \ref{s:System_model} presents the overall system model. In Sections \ref{s:Infinite_buffer} and \ref{s:Finite_buffer}, we study the limiting distribution of the stored energy for infinite- and finite-capacity energy buffers, respectively. In Section \ref{s:BER_outage_analysis}, we analyze the AER and the outage probability of the communication link, when both UL and DL channels are Rayleigh faded. Numerical and simulation results are provided in Section \ref{s:Simulations}. Finally, Section \ref{s:conclusion} concludes the paper. 
\section{System Model}
\label{s:System_model}
We consider a time-slotted point-to-point single-antenna EH system with DL WPT and UL WIT. In particular, the system consists of a node with an EH module which captures the RF energy transferred by an AP in the DL and uses the harvested energy to transmit its backlogged data in the UL. The considered system employs frequency-division-duplex, where WPT and WIT take place concurrently on two different frequency bands. The AP and the EH node are assumed to have no instantaneous knowledge of the DL and the UL CSI, respectively, nor of the amount of harvested energy. Next, we describe the communication, EH, and storage models as well as the considered system imperfections.
\subsection{Communication Model}
In time slot $i$ (defined as the time interval $[i,i+1)$\footnote{The time slot is assumed to be of unit length. Hence, we use the terms energy and power interchangeably.}), the EH node transmits data to the AP with an UL power given by
\begin{equation}
\Pul(i)=\min(B(i),M),
\label{eq:Pul_policy}
\end{equation}
where $B(i)$ is the residual stored energy at the beginning of time slot $i$ and $M$ is the desired constant UL transmit power. The transmitted signal encounters a flat block fading channel, i.e., the channel remains constant over one time slot, and changes independently from one slot to the next. The channel power gain sequence $\{\hul(i)\}$ is a stationary and ergodic process with mean $\Omegaul=\E[\hul(i)]$, where $\E[\cdot]$ denotes expectation. Additive white Gaussian noise (AWGN) with variance $\sigma_n^2$ impairs the received signal at the AP.
\subsection{EH Model}
During the same time slot, the EH node collects $X(i)$ units of RF energy broadcasted by the AP and stores it in its energy buffer.  We assume that the energy replenished in a time slot may only be used in future time slots. The DL channel is also assumed to be flat block fading with  a stationary and ergodic channel power gain sequence $\{\hdl(i)\}$,  assumed to be unknown at the AP, where $\Omegadl=\E[\hdl(i)]$. We adopt the EH receiver model in \cite{WIPT_Architecture_Rui_Zhang_2012}, where the harvested energy in time slot $i$ is given by $X(i)=\eta\Pdl\hdl(i)$, where $0<\eta< 1$ is the RF-to-DC conversion efficiency of the EH module and $\Pdl$ is the constant DL transmit power from the AP. The energy replenishment sequence $\{X(i)\}$ is consequently an independent and identically distributed (i.i.d.) stationary and ergodic process with mean $\bar{X}=\eta\Pdl\Omegadl$, probability density function (pdf) $f(x)$, and complementary cumulative distribution function (ccdf) $\bar{F}(x)=\mathbb{P}(X(i)>x)$, where $\mathbb{P}(\cdot)$ denotes the probability of an event.
\subsection{Storage Model}
\label{ss:storage_model}
The harvested energy $X(i)$ is stored in an energy buffer, such as a rechargeable battery and/or a supercapacitor \cite{Culler_2005}, with a finite storage capacity of $K$. The dynamics of the storage process $\{B(i)\}$ are given by the storage equation
\begin{equation}
\begin{aligned}
B(i+1)&=\min\left(B(i)-\Pul(i)+X(i),K\right)\\
&=\min\left([B(i)-M]^++X(i),K\right) &
\end{aligned},
\label{eq:general_storage_equation}
\end{equation}
where $[x]^+=\max(x,0)$. The storage process $\{B(i)\}$  in (\ref{eq:general_storage_equation}) is a discrete-time Markov chain on a continuous state space $S$, where  $S=[0,K]$ for a finite-size energy buffer and $S=[0,\infty)$ for an infinite-size buffer. \begin{remark}
Interestingly, our storage model is similar to the dam model proposed by Moran in \cite{Moran_1956}. In Moran's model, every year $X(i)$ units of water flow into a dam of capacity $K$ and a constant amount of water $M$ is released just before the following year. Moran studies the amount of water $\{Z(i)\}$ stored in the dam just after release, which is modeled by the storage equation $Z(i+1)= [\min(Z(i)+X(i),K)-M]^+$, which for an infinite-capacity dam (i.e., $K\to\infty$) reduces to\vspace{-0.1cm} 
\begin{equation} 
Z(i+1)=\begin{cases} 0 & Z(i)+X(i) \leq M \\ Z(i)+X(i)-M & Z(i)+X(i)> M \end{cases}.
\label{eq:Moran_buffer_content}
\vspace{-0.1cm} 
\end{equation} 
The stationary distribution\footnote{A stationary distribution of a Markov chain is a distribution such that if the chain starts with this distribution, it remains in that distribution.} of $\{Z(i)\}$ (if it exists) can be obtained e.g. by first defining a new process $U(i)=Z(i)+X(i)$. Hence, if we add $X(i+1)$ to $Z(i+1)$, we get \vspace{-0.1cm}
\begin{equation}
U(i+1)=\begin{cases} X(i+1) & U(i) \leq M \\ U(i)-M+X(i+1) & U(i)> M \end{cases},
\label{eq:U_i_process}
\vspace{-0.1cm} 
\end{equation} 
which is identical in distribution to $\{B(i)\}$ in (\ref{eq:general_storage_equation}) at $K\to\infty$. Hence, the distribution of $\{U(i)\}$ in Moran's dam model is identical to the distribution of $\{B(i)\}$ in our energy buffer model. Similarly, for a finite storage capacity, $\{B(i)\}$ is equivalent to $\{\min(U(i),K)\}$.
\label{remark:equivalence_to_Morans_Model}
\end{remark}
\subsection{Consideration of Imperfections}
\label{s:imperfections}
We consider imperfections due to the circuit power consumption of the EH node and the non-idealities of the energy buffer. In particular, we consider the following imperfections; (a) For the power amplifier of the EH node to transmit an RF power of $\Pul$, it consumes a total power of $\alpha\Pul$, where $\alpha> 1$ is the power amplifier inefficiency. (b) We assume that the EH node circuitry consumes a constant power of $P_{{\rm ct}}$ used mainly for harvesting, processing, and sensing (for EH sensors). (c) Two main imperfections of the energy buffer are considered \cite{Kansal_2007}. First, the buffer is assumed to leak a constant amount of energy in each time slot, denoted by $P_{l}$. Second, we consider the buffer storage inefficiency characterized by $0<\beta<1$, where if $X$ amount of energy is applied at the input of the buffer, only an amount of $\beta X$ may be stored. Compared to rechargeable batteries, supercapacitors have high storage efficiency $\beta$, but also high leakage current  \cite{Culler_2005},\cite{Kansal_2007}.

Define $\Pc$ as the total constant energy usage in each time slot, i.e.,  $\Pc=P_{\rm ct}+P_l$. In this case, the energy buffer dynamics are described by $B(i\!+\!1)\!=\!\min\left(B(i)-(\Pc+\alpha\Pul(i))\!+\!\beta X(i),K\right)$, where the desired UL transmit power is $M$. If $B(i)<\Pc+\alpha M$, then the UL power is reduced to satisfy $B(i)=\Pc+\alpha\Pul(i)$, i.e., $\Pul(i)=(B(i)-\Pc)/\alpha$ which ensures energy neutral operation\footnote{Another option to ensure energy neutral operation is to reduce the circuit power consumption using dynamic voltage scaling or duty cycling \cite{Kansal_2007}.}. Hence, the storage equation reduces to
\begin{equation}
B(i+1)=\min\left([B(i)-(\Pc+\alpha M)]^++\beta X(i),K\right).
\label{eq:storage_equation_imperfection2}
\end{equation}
Observe that (\ref{eq:storage_equation_imperfection2}) is identical to (\ref{eq:general_storage_equation}) after replacing $M$ by $\tilde{M}\!=\!\Pc+\alpha M$ and $f(x)$ by $\tilde{f}(x)\!=\!\frac{1}{\beta}f\left(\frac{x}{\beta}\right)$. Thus, in the following, we perform the analysis for an ideal system (i.e., $\alpha\!=\!1$, $\beta\!=\!1$, and $\Pc\!=\!0$). For a non-ideal system, all the results in Sections \ref{s:Infinite_buffer}-\ref{s:BER_outage_analysis} hold with the aforementioned substitutions.
\section{Infinite-Capacity Energy Buffer}
\label{s:Infinite_buffer}
In this section, we study the energy storage process in (\ref{eq:general_storage_equation}) for an infinite-capacity energy buffer. We provide conditions for which the convergence to a limiting distribution\footnote{A limiting distribution of a Markov chain is a stationary distribution that the chain converges asymptotically to from some initial distribution.} of the buffer content is either guaranteed or violated. Furthermore, we provide the limiting distribution of the buffer content in closed form when the EH process $\{X(i)\}$ is i.i.d. exponentially distributed,  i.e., for a Rayleigh block fading DL channel.
\begin{theorem}\normalfont
For the storage process $\{B(i)\}$ in (\ref{eq:general_storage_equation}) with infinite buffer size, if $M<\bar{X}$, then $\{B(i)\}$ does not possess a stationary distribution. Furthermore, after a finite number of time slots, $\Pul(i)=M$ holds almost surely (a.s).
\label{theo:no_stationary_dist}
\end{theorem}
\begin{proof} The proof is provided in Appendix \ref{app:no_stationary_dist}. \end{proof}
\begin{theorem}\normalfont
For the storage process $\{B(i)\}$ in (\ref{eq:general_storage_equation}) with infinite buffer size, if $M>\bar{X}$, then $\{B(i)\}$ is a stationary and ergodic process which possesses a unique stationary distribution $\pi$ that is absolutely continuous on $(0,\infty)$. Furthermore, the process converges in total variation to the limiting distribution $\pi$ from any initial distribution. 
\label{theo:stationary_dist_infinite}
\end{theorem}
\begin{proof} The proof is provided in Appendix \ref{app:stationary_dist_infinite}. \end{proof}
\begin{theorem}\normalfont
Consider the storage process $\{B(i)\}$ in (\ref{eq:general_storage_equation}) with infinite buffer size and $M>\bar{X}$. Let $g(x)$ on $(0,\infty)$ be the limiting pdf of the energy buffer content, then $g(x)$ must satisfy the following integral equation\vspace{-0.1cm}
\begin{equation}
g(x)=f(x)\int\limits_0^M g(u) \dd u + \int\limits_M^{M+x} f(x-u+M) g(u) \dd u.\vspace{-0.5cm}
\label{eq:Integral_eqn_infinite}
\end{equation}
\label{theo:integral_eq_infinite}
\end{theorem}
\begin{proof}
To understand the integral equation in (\ref{eq:Integral_eqn_infinite}), one may set $B(i)=u$ and $B(i+1)=x$, then (\ref{eq:general_storage_equation}) reads\vspace{-0.1cm}
\begin{equation}
x=\begin{cases} X(i) & u\leq M \\ u-M+X(i) & u>M \end{cases}.\vspace{-0.1cm}
\label{eq:Storage_eq_Infinite_u_x}
\end{equation}
Thus, $g(x|u\leq M)=f(x)$ and $g(x|u> M)=f(x-u+M)$ which is non-zero only for a non-negative amount of harvested energy, i.e., $x-u+\!M\!\geq\!0$. These considerations lead to (\ref{eq:Integral_eqn_infinite}). From the analogy between our storage model and Moran's model, c.f. Remark \ref{remark:equivalence_to_Morans_Model}, (\ref{eq:Integral_eqn_infinite}) is identical to \cite[eq. (5)]{Infinite_dam_Gani_Prabhu_1957}. \vspace{-0.3cm}
\end{proof}
Next, we consider the case when the DL channel is Rayleigh block fading and provide the limiting distribution of the energy buffer content in the following corollary.
\begin{corollary} \normalfont
Consider the storage process in (\ref{eq:general_storage_equation}) with infinite buffer size and $M\!>\!\bar{X}$. If the EH process is exponentially distributed with pdf $f(x)\!=\!\lambda \e^{-\lambda x}$, where $\lambda\!\!=\!\!\frac{1}{\bar{X}}$  and $\delta\!=\!\lambda M\!=\!\frac{M}{\bar{X}}$, then the limiting pdf of the energy buffer content is $g(x)\!\!=\!\!-p\e^{p x}$, where $p\!<\!0$ is given by $p\!=\frac{-\delta-W_{0}(-\delta\e^{-\delta})}{M}$ and $W_0(\cdot)$ is the Lambert W function of order zero.
\label{theo:stationary_dist_infinite_exp}
\end{corollary} 
\begin{proof} The proof is provided in Appendix \ref{app:stationary_dist_infinite_exp}. \end{proof}
\section{Finite-Capacity Energy Buffer}
\label{s:Finite_buffer}
In this section, we first provide the integral equation of the stationary distribution of the storage process $\{B(i)\}$ for a finite-size energy buffer and a general i.i.d. EH process. Then, the distribution is provided for a Rayleigh fading DL channel.
\begin{theorem}  \normalfont
The storage process in (\ref{eq:general_storage_equation}), with a finite buffer size $K$, and an EH process $\{X(i)\}$, which is characterized by a distribution with an infinite positive tail, is a stationary and ergodic process which possesses a unique stationary distribution $\pi$ that has a density on $(0,K)$ and an atom at $K$. Furthermore, the process converges in total variation to the limiting distribution $\pi$ from any initial distribution. 
\label{theo:limiting_dist_Finite}
\end{theorem}
\begin{proof} The proof is provided in Appendix \ref{app:limiting_dist_Finite}. \end{proof}
\begin{theorem}  \normalfont
Consider the storage process $\{B(i)\}$ in (\ref{eq:general_storage_equation}), with a finite buffer size $K$. Let $g(x)$ be the limiting pdf of the energy buffer content on $(0,K)$ and $\pi(K)$ be the limiting probability of a full buffer (i.e., the atom at $K$). If $f(x)$ and $\bar{F}(x)$ are respectively the pdf and the ccdf of $\{X(i)\}$, then, $g(x)$ and $\pi(K)$ must jointly satisfy\vspace{-0.1cm}
\begin{numcases}{g(x)\!=\!\! \label{eq:g_integral_eqn_finite}}
\hspace{-0.1cm}f(x)\int\limits_{u=0}^{M} g(u) \dd u + \int\limits_{u=M}^{M+x} f(x-u+M) g(u)\dd u &  \nonumber \\ \vspace{-0.2cm}
\hspace{4.4cm}0\leq x<K-M & \label{eq:parta}\\
\hspace{-0.1cm}f(x)\int\limits_{u=0}^{M} g(u) \dd u + \int\limits_{u=M}^{K} f(x-u+M) g(u)\dd u & \nonumber \\
+\pi(K) f(x-K+M)  \hspace{1.1cm} K-M \leq x<K & \label{eq:partb}
\end{numcases}
\begin{equation}\vspace{-0.2cm}
\pi(K)\!=\!\frac{\left[\bar{F}(K)\int\limits_{u=0}^{M} g(u) \dd u + \!\!\int\limits_{u=M}^{K} \bar{F}(K-u+M) g(u)\dd u \right]}{1-\bar{F}(M)},\vspace{-0.1cm}
\label{eq:partc}
\end{equation}
and the unit area condition\vspace{-0.1cm}
\begin{equation}
\int\limits_{0}^{K} g(u) \dd u + \pi(K)=1.\vspace{-0.1cm}
\label{eq:unit_area_eq}
\end{equation}
\end{theorem}
\begin{proof}
The integral equations in (\ref{eq:g_integral_eqn_finite}), (\ref{eq:partc}) can be understood by adopting the same approach used to prove (\ref{eq:Integral_eqn_infinite}). In particular, if we set $B(i)=u$ and $B(i+1)=x$, then (\ref{eq:general_storage_equation}) reads
\begin{equation}
x=
\begin{cases}
X(i) & u\leq M \quad\&\quad X(i)<K\\
u-M+X(i) & u>M \quad\&\quad u-M+X(i)<K\\
K & {\rm otherwise.}
\end{cases}\vspace{-0.1cm}
\label{eq:Our_storage_equation_cases}
\end{equation}
Consider first the continuous part of the distribution, i.e., $g(x)$ defined on $0\leq x <K$ given in (\ref{eq:g_integral_eqn_finite}). Eq. (\ref{eq:parta}) is identical to (\ref{eq:Integral_eqn_infinite}), however, we need to further ensure that the upper limit on $u$ given by $M+x$ (for a non-negative harvested energy) is in the domain of $g(u)$, i.e., $M+x<K$ must hold. Hence, (\ref{eq:parta}) is valid only for $x<K-M$ (with strict inequality). For the rest of the range of $x$ in (\ref{eq:partb}), i.e., $K-M\leq x<K$, the upper limit  $M+x$ on $u$ is larger than or equal to $K$. Thus, the whole range of $0< u\leq K$ contributes to $g(x)$. The range $0< u<K$ is covered by the first two integrals in (\ref{eq:partb}), and $u=K$ is considered in the last term. Finally, at $x=K$, the full buffer probability $\pi(K)$ in (\ref{eq:partc}) is obtained similar to (\ref{eq:partb}). However, rather than considering the pdf at the amount of harvested energy $x-[u-M]^+$ as in (\ref{eq:partb}), we consider the ccdf $\bar{F}(x-[u-M]^+)$ instead (at $x\!=\!K$). This is because the full buffer level $K$ is attained when the amount of harvested energy is larger than or equal to $K\!-\![u\!-\!M]^+$, where we sweep over $0\!<\!u\!\leq\!K$ to obtain (\ref{eq:partc}). This completes the proof.\vspace{-0.3cm}
\end{proof}
Next, we consider the case when the DL channel is Rayleigh block fading. We provide the exact limiting distribution of the energy buffer content in Corollary \ref{theo:limiting_dist_Finite_exponential_exact} and an exponential approximation of it in Proposition \ref{prop:limiting_dist_Finite_exponential_approx}.
\begin{corollary}  \normalfont
Consider the storage process $\{B(i)\}$ in (\ref{eq:general_storage_equation}) with a finite buffer size $K$ and an i.i.d. exponentially distributed EH process $\{X(i)\}$ with pdf $f(x)=\lambda\e^{-\lambda x}$, where $\lambda=\frac{1}{\bar{X}}$ and $\delta=\lambda M$, then the limiting pdf $g(x)$ of the energy buffer content and the full buffer probability $\pi(K)$ are given by
 \small
\begin{equation}
\begin{aligned}
&g(x)=\pi(K)\lambda\e^{-\lambda(x-K)} \Bigg[1+\sum\limits_{q=1}^{n}\frac{\e^{-\delta q}}{(q-1)!} \left(\delta q+\lambda(x-K)\right)^{q-1}\\
&\left(\frac{\lambda(x-K)}{q}+\delta-1\right)\Bigg], \hspace{0.5cm}\begin{aligned}&\\&[K-(n+1)M]^+\leq x<K-nM,\\ 
																																										&\hspace{2cm} n=0,\ldots,l', \end{aligned}
\end{aligned}
\label{eq:g_x_finite_exact}
\end{equation}
\normalsize \vspace{-0.1cm}
and \vspace{-0.2cm}
 \small
\begin{equation}
\begin{aligned}
&\pi(K)\!=\!\!\Bigg\{\!\sum\limits_{n=0}^{l-1}\e^{n\delta}\!\left(\!\e^{\delta}\!-\!\!1\!\!+\!\!\sum\limits_{q=1}^{n}\!\frac{\left(\delta\e^{-\delta}\right)^q}{q!}\!\left(\!\e^{\delta}\left(q\!-\!(n\!+\!1)\right)^q\!-\!(q\!-\!n)^q\right)\!\right) \\
&+\!\e^{l\delta}\!\left(\e^{\lambda\Delta}\!\!-\!1\!+\!\sum\limits_{q=1}^{l}\!\frac{\left(\delta\e^{-\delta}\right)^q}{q!}\!\left(\e^{\lambda\Delta}\!\left(\!q\!-\! \frac{K}{M}\right)^q\!-\!(q\!-\!l)^q\right)\right)\!+\!1\Bigg\}^{-1},
\end{aligned}
\label{eq:Pi_K_delta}
\end{equation}
\normalsize
where $K=lM+\Delta$ with $l\in\mathbb{Z}$ and $0\leq\Delta < M$. In (\ref{eq:g_x_finite_exact}), $l'$ is either $l'=l-1$ if $\Delta=0$ or $l'=l$ if $\Delta\neq 0$.
\label{theo:limiting_dist_Finite_exponential_exact}  
\end{corollary}
\begin{proof} 
First, we note that the solution of $g(x)$ in (\ref{eq:g_x_finite_exact}) is obtained in stripes of width $M$. This is due to the upper integral limit $M+x$ in (\ref{eq:parta}), hence the width-$M$ stripes solution is in fact general for any distribution of the i.i.d. EH process \cite{Moran_book_1961}. We derived $g(x)$ by induction. In particular, we obtained $g(x)$ in the range $K-M\leq x <K$ from (\ref{eq:partb}) and (\ref{eq:partc}), after setting $f(x)=\lambda\e^{-\lambda x}$ and $\bar{F}(x)=\e^{-\lambda x}$. Then, using (\ref{eq:parta}) and (\ref{eq:partb}), we traversed back in sections of width $M$ until $x=0$. Note that the provided solution is general for any $K$ (i.e., $K$ is not necessarily an integer multiple of $M$). The exact derivation is lengthy so we omit it and provide it in the journal version of this paper. However, we note that due to the analogy between our storage model and Moran's dam model, c.f. Remark \ref{remark:equivalence_to_Morans_Model}, the solution of $g(x)$ is identical to \cite[eq. (3.6)]{Moran_book_1961} (which is also given by Prabhu in \cite[Section 2]{Prabhu_1958}). After getting $g(x)$, $\pi(K)$ is obtained by solving (\ref{eq:unit_area_eq}).\vspace{-0.2cm}
\end{proof}
Since the exact limiting distribution of the buffer content provided in Corollary \ref{theo:limiting_dist_Finite_exponential_exact} is quite complicated, we propose an exponential-type approximation which will be used in the AER and the outage probability analysis in Section \ref{s:BER_outage_analysis}. 
\begin{proposition}  \normalfont
First for notational brevity, we define the $n^{{\rm th}}$ section of $g(x)$ in (\ref{eq:g_x_finite_exact}) as $g_n(x)=g(x),\,[K-(n+1)M]^+\leq x<K-nM$. The limiting distribution of the storage process described in Corollary \ref{theo:limiting_dist_Finite_exponential_exact} can be approximated in the range $0\leq x < K-n_c M$ by $\tilde{g}(x)$, where $n_c$ is some chosen section number after which $g_n(x)\approx \tilde{g}(x),\,\forall n\geq n_c$ as shown in Fig. \ref{fig:pdf_approximation}. For $0\leq n<n_c$, $\tilde{g}_n(x)$ is given by $g_n(x)$ in (\ref{eq:g_x_finite_exact}) after replacing $\pi(K)$ with $\tilde{\pi}(K)$, i.e., $\tilde{g}_n(x)=\frac{\tilde{\pi}(K)}{\pi(K)}g_n(x)$, where $\tilde{\pi}(K)$ is the approximate full buffer probability that ensures a unit area of the approximate distribution. The proposed approximation is tight for $K\geq 3M$ and $n_c\geq 2$.
\vspace{-0.3cm}
\begin{figure}[!h]
\centering
\scalebox{0.59}{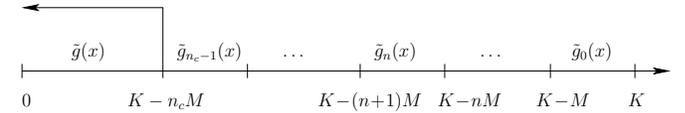}\vspace{-0.3cm}
\caption{Pdf approximation.}
\label{fig:pdf_approximation}
\end{figure}\newline
We propose an exponential-type approximation given by $\tilde{g}(x)=c\e^{dx}$, where $d$ and $c$ are given by \vspace{-0.2cm}
\begin{equation}
d=\frac{-\delta-W_j(-\delta\e^{-\delta})}{M},\quad j=\begin{cases} -1 & 0<\delta\leq1\\
																																		0 & \delta>1
																																\end{cases},
\label{eq:d_approx}
\end{equation} \vspace{-0.2cm}
\small
\begin{equation*}
c=\tilde{\pi}(K)\,\lambda\underbrace{\e^{\lambda K} \left[1\!+\!\sum\limits_{q=1}^{l'}\frac{\e^{-\delta q}}{(q-1)!} \left(\delta q\!-\!\lambda K\right)^{q-1}\left(\frac{-\lambda K}{q}\!+\!\delta\!-\!1\right)\right]}_{\Sigma_1},
\label{eq:c_approx}
\end{equation*}
\normalsize \vspace{-0.2cm}
and the approximate atom at $K$ is given by $\tilde{\pi}(K)\!=\!\!$
\small
\begin{equation}
\Bigg[\!1\!+\!\Sigma_2\!+\!\!\sum\limits_{n=0}^{n_c-1}\!\!\e^{n\delta}\Big(\!\e^{\delta}\!-\!1\!+\!\sum\limits_{q=1}^{n}\!\frac{\left(\delta\e^{-\delta}\right)^q}{q!}\!\left(\e^{\delta}\left(q\!-\!(n\!+\!1)\right)^q\!-\!(q\!-\!n)^q\!\right)\!\Big)\!\Bigg]^{-1},
\label{eq:Pi_K_approx}
\end{equation}
\normalsize
where $\Sigma_2=\begin{cases} \frac{\lambda\Sigma_1}{d}\left(\e^{d(K-n_c M)}-1\right) & \delta\neq 1\\
\lambda\Sigma_1\left(K-n_c M\right) & \delta=1 \end{cases}$, and $W_j(\cdot)$ is the $j^{\text{th}}$ order Lambert W function.
\label{prop:limiting_dist_Finite_exponential_approx}
\end{proposition}
\begin{proof} The proof is provided in Appendix \ref{app:limiting_dist_Finite_exponential_approx}.\end{proof}
\section{AER and Outage Probability Analysis}
\label{s:BER_outage_analysis}
In this section, we analyze the AER and the outage probability of the communication over the UL channel, when both UL and DL channels are Rayleigh faded. In the finite-size buffer case, we use the approximate pdf $\tilde{g}(x)\!=\!c\e^{dx}$ given in Proposition \ref{prop:limiting_dist_Finite_exponential_approx}, whereas for an infinite-size buffer with $\delta\!>\!1$, we use the exact pdf $g(x)\!=\!-p\e^{px}$ given in Corollary \ref{theo:stationary_dist_infinite_exp}. Hence, we show only the results of the finite-size buffer and deduce the latter by setting $c\!=\!-p$ and $d\!=\!p$. For an infinite-size buffer with $\delta\leq1$, we use $\Pul(i)\!=\!M$, $\forall i$, c.f. Theorem \ref{theo:no_stationary_dist}.\vspace{-0.1cm}
\subsection{AER Analysis}
\label{ss:BER_Analysis}
For uncoded transmission, the bit or symbol error rate of many modulation schemes can be expressed as $P_e(\gamma)=aQ(\sqrt{b\gamma})$ \cite{Wang_Giannakis_2003}, where $\gamma$ is the instantaneous signal-to-noise-ratio (SNR), $Q(\cdot)$ is the Gaussian Q-function, and $a$, $b$ depend on the modulation scheme used, e.g., for binary phase shift keying (BPSK) $a\!=\!1$ and $b\!=\!2$.

For an infinite-size buffer with $\delta\leq 1$, the AER is given by \vspace{-0.1cm}
\begin{equation}
P_e\big|_{\infty,\delta\leq 1}\!=\!\!\int\limits_0^\infty\!a Q\left(\!\sqrt{\frac{bM\Omegaul h}{\sigma_n^2}}\right)\!\!\e^{-h} \dd h\!=\!\frac{a}{2}\left[\!1\!-\!\sqrt{\frac{b\bar{\gamma}\delta}{2+b\bar{\gamma}\delta}}\right],\vspace{-0.1cm}
\label{eq:BER_infinite}
\end{equation}
where $\bar{\gamma}$ is defined as $\bar{\gamma}=\Omegaul \bar{X}/\sigma_n^2=\Omegaul/\left(\lambda \sigma_n^2\right)$.
For a finite-size buffer, the AER is given by\vspace{-0.2cm}
\begin{equation}
P_e\big|_{\rm F}\!=\!\int\limits_0^M\!\int\limits_0^\infty a Q\left(\sqrt{\frac{bx\Omegaul h}{\sigma_n^2}}\right) \e^{-h} \dd h\, c\e^{dx} \dd x+P_MP_e\big|_{\infty,\delta\leq 1}, \vspace{-0.1cm}
\label{eq:BER_finite_step1}
\end{equation}
where we define $P_M\!=\!\mathbb{P}(\Pul(i)\!=\!M)=\mathbb{P}(B(i)\geq M)=1-\int_0^M c\e^{dx} \dd x=1-\frac{c}{\lambda}$, where we used $\lambda \e^{dM}=\lambda+d$, c.f. Appendix \ref{app:stationary_dist_infinite_exp}. The first term in (\ref{eq:BER_finite_step1}) can be simplified to $\frac{ a c}{2\lambda}\int_0^{\delta}\e^{\frac{d}{\lambda}x}\left(1-\sqrt{\frac{b\bar{\gamma}x}{2+b\bar{\gamma}x}}\right) \dd x$, where $\int_0^{\delta}\e^{\frac{d}{\lambda}x}\dd x=\left(\e^{dM}-1\right)/(d/\lambda)=1$. Substituting back in (\ref{eq:BER_finite_step1}), we get\vspace{-0.2cm}
\begin{equation}
P_e\big|_{\rm F}\!=\!\frac{a}{2}\!\Bigg[1\!-\!\left(\!1\!-\!\frac{c}{\lambda}\right)\sqrt{\!\frac{b\bar{\gamma}\delta}{2\!+\!b\bar{\gamma}\delta}}-\frac{c}{\lambda}\!\int\limits_0^{\delta}\!\!\sqrt{\!\frac{b\bar{\gamma}x}{2\!+\!b\bar{\gamma}x}} \e^{\frac{d}{\lambda}x}\dd x\Bigg],\vspace{-0.2cm}
\label{eq:BER_finite}
\end{equation}
where the integral in (\ref{eq:BER_finite}) has finite limits and can be solved numerically.

In order to study the diversity order of the AER in (\ref{eq:BER_infinite}) and (\ref{eq:BER_finite}), we consider the high SNR regime, i.e., as $\bar{\gamma}\to\infty$. We first note that $\lim_{y\to\infty}\left(1-\!\sqrt{\frac{y}{2+y}}\right)\!=\!\frac{1}{y}$. Hence, the AER in (\ref{eq:BER_infinite}) tends asymptotically (denoted by ``$\asymp$") to 
$P_e\big|_{\infty,\delta\leq 1}\asymp \frac{a}{2b\delta\bar{\gamma}}$. That is, an infinite-size buffer with $\delta\leq1$ achieves a diversity order of 1 with respect to the AER. For a finite-size buffer, the first term in (\ref{eq:BER_finite_step1}), given by $\frac{ a c}{2\lambda}\int_0^{\delta}\e^{\frac{d}{\lambda}x}\left(1-\sqrt{\frac{b\bar{\gamma}x}{2+b\bar{\gamma}x}}\right) \dd x$, tends asymptotically to $\frac{ a c}{2\lambda}\int_0^{\delta}\e^{\frac{d}{\lambda}x}\frac{1}{b\bar{\gamma}x} \dd x$, which also has a diversity order of 1. Hence, the diversity order is not affected by the finite capacity of the energy buffer. Therefore, (\ref{eq:BER_finite}) tends asymptotically to\vspace{-0.2cm}
\begin{equation}\hspace{-0.2cm}
P_e\big|_{\rm F}\!\asymp\!\frac{ a c}{2\lambda b\bar{\gamma}}\int\limits_0^{\delta}\frac{\e^{\frac{d}{\lambda}x}}{x} \dd x+ \left(1-\frac{c}{\lambda}\right) \frac{a}{2b\delta\bar{\gamma}}.\vspace{-0.1cm}
\label{eq:BER_finite_asymp}
\end{equation}
\subsection{Outage Probability Analysis}
\label{ss:Outage_probability_analysis}
Since the CSI is unknown at the EH node, the node transmits data at a constant rate $R_0$ in bits/(channel use). Therefore, assuming a capacity-achieving code, an outage occurs whenever $R_0>\log_2(1+\gamma)\Rightarrow \gamma<\gamma_{\rm thr}$, where $\gamma$ is the UL instantaneous SNR and $\gamma_{\rm thr}=2^{R_0}-1$. Hence, the outage probability for an infinite-size buffer with $\delta\leq 1$ is\vspace{-0.1cm}
\begin{equation}
P_{\text{out}}\big|_{\infty,\delta\leq 1}\!=\!\mathbb{P}\left(\gamma<\gamma_{\text{thr}}\right)\!=\!\mathbb{P}\left(\frac{M\Omegaul h}{\sigma_n^2}\!<\!\gamma_{\text{thr}}\right)\!=\!1-\e^{-\frac{\gamma_{\text{thr}}}{\delta\bar{\gamma}}}.
\label{eq:Pout_infinite}\vspace{-0.1cm}
\end{equation}
For a finite-size buffer, the outage probability is given by \vspace{-0.2cm}
\begin{equation}
P_{\text{out}}\big|_{\rm F}\!=\!\int\limits_0^M\!\mathbb{P}\left(\!\frac{x \Omegaul h}{\sigma_n^2}<\gamma_{\text{thr}}\!\right) c\e^{dx} \dd x\!+ \!P_MP_{\text{out}}\big|_{\infty,\delta\leq 1}.\vspace{-0.1cm}
\label{eq:Pout_finite_step}
\end{equation} 
The first term in (\ref{eq:Pout_finite_step}) reduces to $\int_0^\delta  \mathbb{P}\left(x\bar{\gamma}h\!<\!\gamma_{\text{thr}}\right) \frac{c}{\lambda}\e^{\frac{d}{\lambda}x} \dd x=\frac{c}{\lambda}\int_0^\delta\left(1-\e^{-\frac{\gamma_{\text{thr}}}{x\bar{\gamma}}}\right)\e^{\frac{d}{\lambda}x} \dd x$, where $\int_0^{\delta}\e^{\frac{d}{\lambda}x}\dd x=1$. Using $P_M=1-\frac{c}{\lambda}$ and $P_{\text{out}}\big|_{\infty,\delta\leq 1}$ in (\ref{eq:Pout_infinite}), (\ref{eq:Pout_finite_step}) reduces to\vspace{-0.3cm}
\begin{equation}
P_{\text{out}}\big|_{\rm F}\!=\left(1-\e^{-\frac{\gamma_{\text{thr}}}{\bar{\gamma}\delta}}\right) +\frac{c}{\lambda}\left[\e^{-\frac{\gamma_{\text{thr}}}{\bar{\gamma}\delta}}-\int\limits_{0}^{\delta}\e^{-\frac{\gamma_{\text{thr}}}{\bar{\gamma}x}}\e^{\frac{d}{\lambda}x} \dd x\right].\vspace{-0.1cm}
\label{eq:Pout_finite}
\end{equation}
Using $\lim_{y\to\infty}\e^{-\frac{1}{y}}=1-\frac{1}{y}+o(y^{-2})$, the outage probability in (\ref{eq:Pout_infinite}) tends asymptotically to $P_{\text{out}}\big|_{\infty,\delta\leq 1}\asymp\frac{\gamma_{\text{thr}}}{\delta\bar{\gamma}}$, i.e., with a diversity order of 1. For a finite-size buffer, the outage probability in (\ref{eq:Pout_finite}) tends asymptotically to\vspace{-0.2cm}
\begin{equation}
P_{\text{out}}\big|_{\rm F} \asymp  \left(1-\frac{c}{\lambda}\right)\frac{\gamma_{\text{thr}}}{\bar{\gamma}\delta}+\frac{c\gamma_{\text{thr}}}{\lambda\bar{\gamma}} \int\limits_{0}^{\delta}\frac{1}{x}\e^{\frac{d}{\lambda}x} \dd x.\vspace{-0.2cm}
\label{eq:outage_prob_delta_asymp}
\end{equation}
Similar to the AER, the diversity order of the outage probability is unaffected by the finite capacity of the energy buffer. We note that, although for a small buffer size, namely $K\leq 3M$, the approximate pdf on $[0,M]$ is not tight, it can be shown using the exact pdf in (\ref{eq:g_x_finite_exact}) that in this case the diversity order of the AER and the outage probability is still 1.
\section{Simulation and Numerical Results}
\label{s:Simulations}
In this section, we evaluate the performance of the investigated energy management policy through simulations. The simulation parameters are listed in Table \ref{tab:simulation_parameters}. 
\begin{table}[!tp]
\caption{System Parameters}
\begin{tabular}{@{}ll@{}}  \toprule   
Parameter & Value \\ \midrule \addlinespace[0.5em]
AP to EH node distance & $5\,$m\\ \addlinespace[0.5em]
AP and EH node antenna gains & $12\,$dBi and $2\,$dBi\\ \addlinespace[0.5em]
DL transmit power & $\Pdl=1\,$W\\ \addlinespace[0.5em]
AP noise figure & $5\,$dB \\ \addlinespace[0.5em]
Path loss exponent of DL and UL channels & $2.7$ \\ \addlinespace[0.5em]
DL and UL channel models & Rayleigh block fading \\ \addlinespace[0.5em]
DL and UL center frequencies  & $915\,$MHz and $2.45\,$GHz\\ \addlinespace[0.5em]
Power amplifier inefficiency & $\alpha=1.5$ \\ \addlinespace[0.5em]
Storage efficiency & $\beta=0.9$ \\ \addlinespace[0.5em]
Average harvested energy & $\widetilde{\bar{X}}=\beta\bar{X}=10^{-5}\,$J\\ \addlinespace[0.5em]
RF-to-DC conversion efficiency & $\eta=0.7$\\ \addlinespace[0.5em]
Total constant power consumption & $\Pc=0.2\,\mu$W\\ \addlinespace[0.5em]
Storage capacity & $K\!=\!4\widetilde{\bar{X}}$, $7\widetilde{\bar{X}}$, and $20\widetilde{\bar{X}}$\\ \addlinespace[0.5em]
\bottomrule
\end{tabular}
\label{tab:simulation_parameters}
\end{table}

Fig. \ref{fig:BER_Rayleigh_finite_battery_closed_form_K_Eavg} shows the AER of the received signal at the AP when the EH node transmits a BPSK signal over a bandwidth of $BW=5\,$MHz. At room temperature ($300\,$K), this corresponds to a noise power of $-103\,$dBm at the AP and an SNR of $\widetilde{\bar{\gamma}}=\Omegaul \widetilde{\bar{X}}/\sigma_n^2=24.6\,$dB, where $\widetilde{\bar{X}}$ is given in Table \ref{tab:simulation_parameters}. We sweep over $\tilde{\delta}=\tilde{M}/\widetilde{\bar{X}}=0.1,\ldots,1.7$, which corresponds to a desired UL transmit power of $M=0.6\,\mu{\rm W},\cdots,12\,\mu$W. The closed-form results shown in Fig. \ref{fig:BER_Rayleigh_finite_battery_closed_form_K_Eavg} are obtained from the expressions in Section \ref{ss:BER_Analysis}, where for a finite-storage capacity, we use (\ref{eq:BER_finite}), and for an infinite-storage capacity, we use (\ref{eq:BER_infinite}) for $\tilde{\delta}\leq1$ and (\ref{eq:BER_finite}) for $\tilde{\delta}>1$, with $c=-p$ and $d=p$, c.f. Corollary \ref{theo:stationary_dist_infinite_exp}. We observe that the closed-form results agree perfectly with the simulated results. This emphasizes the tightness of the approximate pdf provided in Proposition  \ref{prop:limiting_dist_Finite_exponential_approx}. It is observed that for the considered energy management scheme, the optimal $\tilde{\delta}$, for which the AER is minimized, is always $\leq1$ and  increases with the storage capacity\footnote{We note that while instantaneous CSI knowledge is not required for the adopted transmission protocol, statistical CSI is needed if $\tilde{\delta}$ is to be optimized.}. As $K\to\infty$, the optimal $\tilde{\delta}\to1$. In other words, for the considered energy transmission policy, the optimal desired UL transmit power is higher for larger energy buffers, but it is always less than the average harvested power. Furthermore, our results show that, for a given storage capacity, the optimal desired UL power decreases with the SNR. This result is not shown here due to space limitations. 
\begin{figure}[!tp]
\centering
\includegraphics[width=0.39\textwidth]{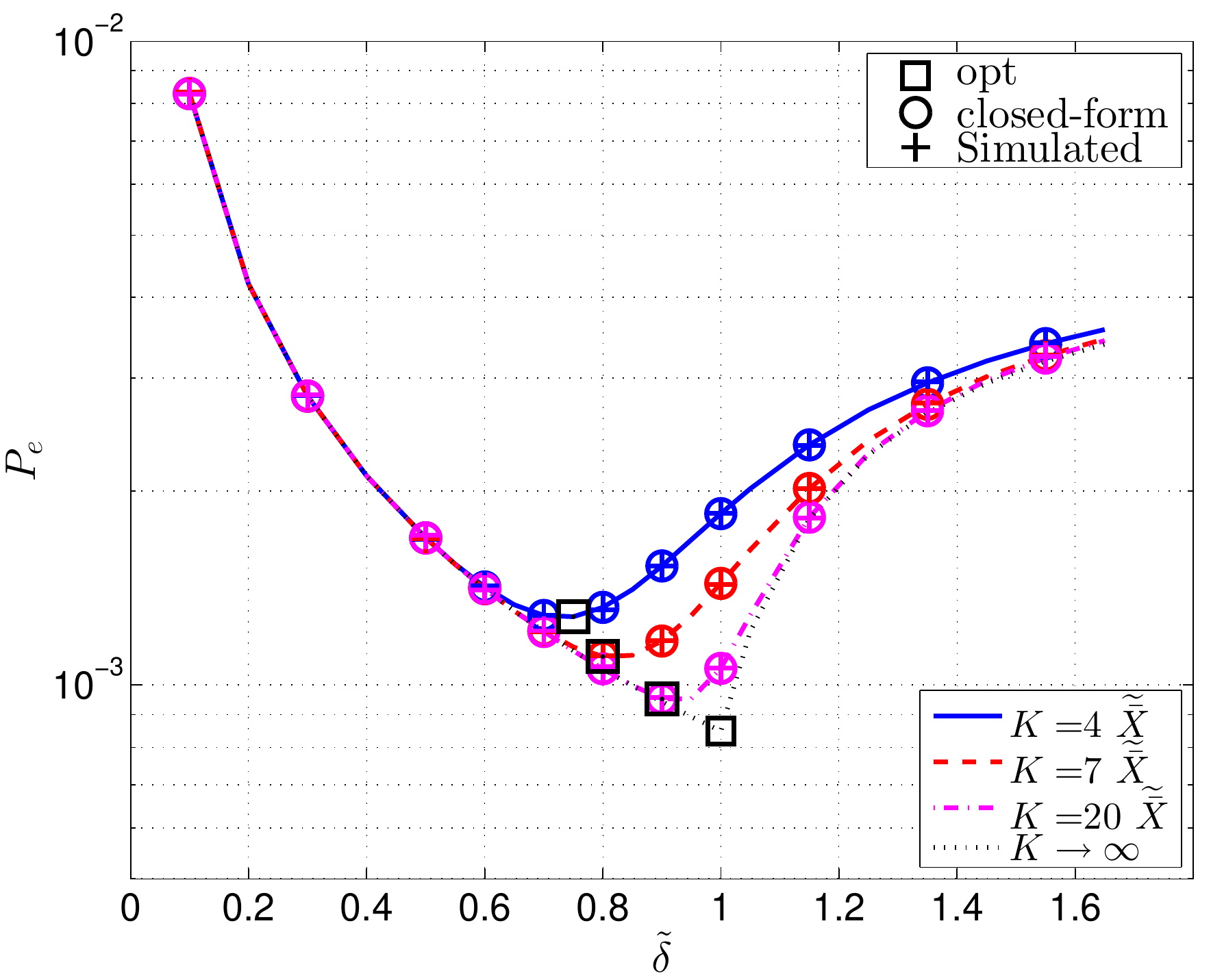}\vspace{-0.3cm}
\caption{AER for different buffer sizes and different desired UL transmit power.}
\label{fig:BER_Rayleigh_finite_battery_closed_form_K_Eavg}
\end{figure}

Fig. \ref{fig:Outage_Rayleigh_finite_battery_vs_SNR} shows the outage probability of the UL channel when the EH node transmits at a constant rate of $2.0574\,$bits/(channel use), i.e., for $\gamma_{\rm thr}=5\,$dB. We sweep over different SNRs of $\widetilde{\bar{\gamma}}=10,\ldots,40\,$dB, which corresponds to an UL channel bandwidth range of $BW=\frac{\Omegaul \widetilde{\bar{X}}}{\widetilde{\bar{\gamma}}K_{\rm B}T_{\rm e}}=147\,$MHz$,\ldots,147\,$KHz, respectively, where $K_{\rm B}$ is Boltzmann's constant and $T_{\rm e}$ is the equivalent noise temperature of the AP. For a given SNR $\widetilde{\bar{\gamma}}$ and a given buffer size $K$, the energy management policy is operated at the optimal $\tilde{\delta}$ for which the outage probability is minimized. The closed-form results shown in Fig. \ref{fig:Outage_Rayleigh_finite_battery_vs_SNR} are obtained from the expressions in Section \ref{ss:Outage_probability_analysis}, where for a finite-storage capacity, we use (\ref{eq:Pout_finite}), and for an infinite-storage capacity, we use (\ref{eq:Pout_infinite}) at $\delta_{\rm opt}=1\,\forall\, \widetilde{\bar{\gamma}}$. Observe that in Fig. \ref{fig:Outage_Rayleigh_finite_battery_vs_SNR}, the outage probability curves for the different energy buffer sizes are parallel. This agrees with our asymptotic analysis in Section \ref{ss:Outage_probability_analysis}, which shows that diversity order is not affected by an energy storage with finite capacity. 
\begin{figure}[!tp]
\centering
\includegraphics[width=0.4\textwidth]{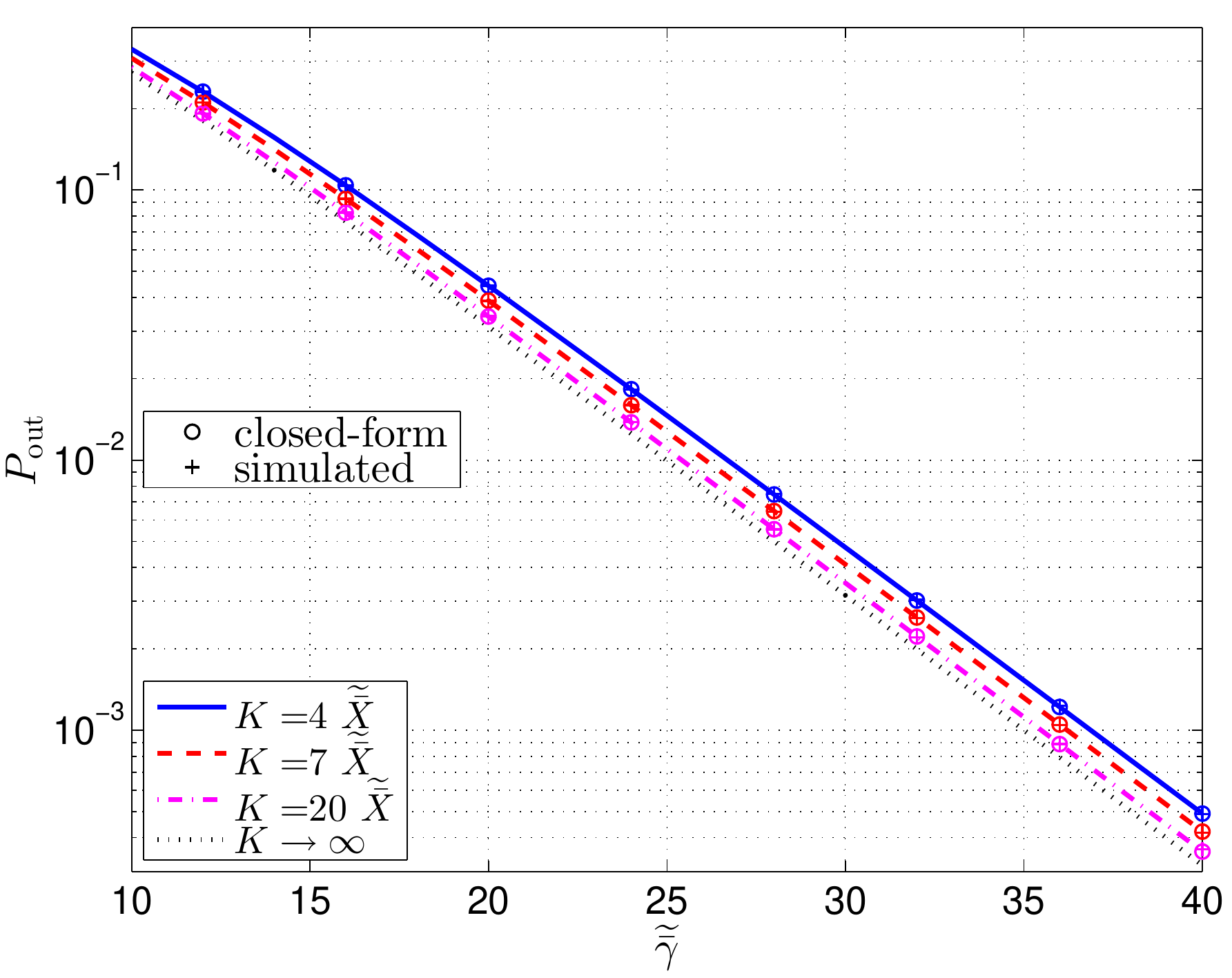}\vspace{-0.3cm}
\caption{Outage probability for different buffer sizes and different SNRs.}
\label{fig:Outage_Rayleigh_finite_battery_vs_SNR}
\end{figure}

\section{Conclusion}
\label{s:conclusion}
In this paper, we considered a simple online energy neutral transmission policy for an EH node, with finite/infinite energy storage. Using the theory of discrete-time Markov chains on a general state space, we analyzed the limiting distribution of the stored energy in the buffer for a general i.i.d. EH process and obtained it in closed form for an exponential EH process. An exponential-type approximation of the stored content distribution is proposed for finite-size buffers and shown to be tight. Our results reveal that the diversity orders of the AER and the outage probability are not affected by a finite energy storage capacity. Furthermore, for the considered transmission scheme, it was shown that the optimal desired transmit power of the EH node is always less than the average harvested power and increases with the storage capacity but decreases with the SNR.\vspace{-0.1cm} 
\appendices  
\section{Proof of Theorem \ref{theo:no_stationary_dist}}
\label{app:no_stationary_dist} \vspace{-0.1cm}
Setting $K\to\infty$ and taking the expectation of both sides of (\ref{eq:general_storage_equation}), we obtain \vspace{-0.2cm}
\begin{equation} 
\E[B(i+1)]-\E[B(i)]=\bar{X}-\E[\Pul(i)].\vspace{-0.2cm}
\label{eq:expection_storage_equation}
\end{equation}
From (\ref{eq:Pul_policy}), $\Pul(i) \leq M, \, \forall \, i \Rightarrow \E[\Pul(i)] \leq M$, hence from (\ref{eq:expection_storage_equation}), $\E[B(i+1)]-\E[B(i)]\geq \bar{X}\!-\!M$ follows.
If $M\!<\!\bar{X}$, then \vspace{-0.2cm}
\begin{equation}
\E[B(i+1)]>\E[B(i)]\vspace{-0.2cm}
\label{eq:accumulating_energy}
\end{equation}
must hold. That is, the mean of the process $\{B(i)\}$ changes (increases) with time, and therefore a stationary distribution for $\{B(i)\}$ does not exist. Furthermore, from (\ref{eq:accumulating_energy}), $\lim_{i\to \infty} \E[B(i)]=\infty$, i.e., the energy accumulates in the buffer. Hence, there must be some time slot $j$, after which for $i>j$, $B(i)> M$ a.s. Next, we prove by contradiction that $j$ must be finite. If $\Pul(j)=B(j)<M$ and $j\to \infty$, then $\lim\limits_{j\to\infty} \E[B(j)]<M$ which violates $\lim\limits_{i\to \infty} \E[B(i)]=\infty$. Hence, $j$ must be finite. This completes the proof.\vspace{-0.1cm}
\section{Proof of theorem \ref{theo:stationary_dist_infinite}}
\label{app:stationary_dist_infinite}
From Remark \ref{remark:equivalence_to_Morans_Model}, it can be observed that Moran's process $\{Z(i)\}$ in (\ref{eq:Moran_buffer_content}) is equivalent to the waiting time of a customer in a GI/G/1 queue \cite{asmussen2003applied}, where $X(i)$ is equivalent to the customer service time and $M$ is equivalent to the customers' inter-arrival time. Now, our storage process $\{B(i)\}$ in (\ref{eq:general_storage_equation}) with $K\!\to\!\infty$ is equivalent to the process $U(i)=Z(i)+X(i)$, see (\ref{eq:U_i_process}). That is, $\{B(i)\}$ is equivalent to the sojourn time (waiting time plus service time) of a customer in a GI/G/1 queue. Since $\{Z(i)\}$ and $\{X(i)\}$ are independent and $\{X(i)\}$ is stationary, then the steady state behavior of $\{B(i)\}$ is solely governed by that of $\{Z(i)\}$. Hence, from \cite[Corollary 6.5 and Corollary 6.6]{asmussen2003applied}, $M>\bar{X}$ is a sufficient condition for the process $\{B(i)\}$ to possess a unique stationary distribution to which it converges in total variation from any initial distribution.
\section{Proof of Corollary \ref{theo:stationary_dist_infinite_exp}}
\label{app:stationary_dist_infinite_exp}
Substituting $f(x)\!=\!\lambda\e^{-\lambda x}$ in (\ref{eq:Integral_eqn_infinite}) and using $\delta\!=\!\lambda M$, we get\vspace{-0.1cm}
\begin{equation}
g(x)=\lambda\e^{-\lambda x}\left[\int\limits_0^M g(u) \dd u + \int\limits_{M}^{M+x}\e^{-\delta}\e^{\lambda u} g(u) \dd u\right].\vspace{-0.1cm}
\label{eq:integral_eqn_INF_battery_exp_input}
\end{equation}
When $M>\bar{X}$, i.e., $\delta>1$, we know from Theorem \ref{theo:stationary_dist_infinite} that (\ref{eq:integral_eqn_INF_battery_exp_input}) has a unique solution for $g(x)$. Similar to \cite[eq. (11)]{Infinite_dam_Gani_Prabhu_1957}, we postulate an exponential-type solution given by $g(x)=k\e^{px}$, then the right hand side of (\ref{eq:integral_eqn_INF_battery_exp_input}) reduces to\vspace{-0.1cm}
\begin{eqnarray}
&\lambda\e^{-\lambda x}\left[\frac{k}{p}\left(\e^{pM}-1\right)+\frac{k\e^{-\delta}}{\lambda+p}\left(\e^{(\lambda+p)(M+x)}-\e^{(\lambda+p)M}\right) \right]\notag\\
&\!=\!\e^{-\lambda x}\!\left[\frac{\lambda k}{p}\left(\e^{pM}\!-\!1\right)\!-\!\frac{\lambda k}{\lambda+p}\e^{pM}\right]\!+\!\frac{k \lambda\e^{pM} }{\lambda +p}\e^{px}\!\req k\e^{px}.
\label{eq:RHS_infinite_exp}
\end{eqnarray} 
In order for (\ref{eq:RHS_infinite_exp}) to hold, the coefficient of $\e^{px}$ in the second term of (\ref{eq:RHS_infinite_exp}) must be $k$, which implies $\lambda\e^{pM}\!=\!\lambda\!+\!p$. This condition will also reduce the coefficient of $\e^{-\lambda x}$ to zero. From $\lambda\e^{pM}\!=\!\lambda+p$, $p$ can be obtained using the Lambert W function, i.e., $p\!=\!\left(-\delta-W_{0}(-\delta\e^{-\delta})\right)/M$, which is $<0$ since $\delta>1$. Now, $k$ can be obtained from the unit area condition on $g(x)$, namely, $\int_0^\infty k\e^{px}\!\!=\!\!1 \Rightarrow k\!=\!-p$. This completes the proof.
\section{Proof of theorem \ref{theo:limiting_dist_Finite}}
\label{app:limiting_dist_Finite} 
Similar to the random walk process on a half line in \cite[Section 4.3.1]{Meyn_Tweedie}, if the distribution of the EH process $\{X(i)\}$ has an infinite positive tail, then the state space $S$ contains an atom at $K$, i.e., the energy level $B(i)\!=\!K$ is reachable with non-zero probability. Define the measure $\phi$ as $\phi(0,K)\!=\!0$ and $\phi(\{K\})\!=\!1$, then the process $\{B(i)\}$ is $\phi$-irreducible, see \cite[Section 4.2]{Meyn_Tweedie}. Furthermore, $\{B(i)\}$ is also $\psi$-irreducible with $\psi(A)=\sum_n \mathbb{P}^n(K,A)2^{-n}$, where $\mathbb{P}^n(x,A)$ is the probability that the Markov chain moves from energy state $x$ to energy set $A$ in $n$ time steps. The dynamics of $\{B(i)\}$ in (\ref{eq:general_storage_equation}) ensures that all energy sets are reachable a.s. from any initial state of the buffer in a finite mean time. Hence, the chain is positive Harris recurrent \cite[Proposition 9.1.1]{Meyn_Tweedie}, where \emph{positive} recurrence follows from \cite[Theorem 10.2.2]{Meyn_Tweedie}. Thus, $\{B(i)\}$ possesses a unique stationary distribution $\pi$. Finally, with the additional property of $\{B(i)\}$ being aperiodic (i.e., no energy level sets are only revisited after a fixed number of time slots $>1$ (period $>1$)), it follows from \cite[Theorem 13.3.3]{Meyn_Tweedie} that $\{B(i)\}$ converges to the distribution $\pi$ in total variation from any initial distribution $\Gamma$, i.e., $\lim\limits_{n\to\infty}\sup\limits_{A}|\int \Gamma(\dd x)\mathbb{P}^n(x,A)-\pi(A)|\to0$. This completes the proof.
\section{Proof of Proposition \ref{prop:limiting_dist_Finite_exponential_approx}}
\label{app:limiting_dist_Finite_exponential_approx}
The exponential approximation is motivated by the exponential distribution of the buffer content for an infinite buffer size given in Corollary \ref{theo:stationary_dist_infinite_exp}. The reason why we approximate only part of $g(x)$ is that although the approximate pdf $\tilde{g}(x)$ is tight for most of the range of $x$ (even for $n_c=2$), it is loose at the tail of the distribution (namely for the last two sections of the pdf, i.e., $n=0,1$). With $\tilde{g}(x)=c\e^{d x}$, $d$ is obtained in exactly the same manner as $p$ for an infinite-size buffer, c.f. Appendix \ref{app:stationary_dist_infinite_exp}. However, unlike in the infinite-size buffer case,  the amount of energy in a finite-size buffer with $\delta \leq 1$ still convergences to a limiting distribution, c.f. Theorem \ref{theo:limiting_dist_Finite}. This explains the use of the Lambert W function with two different orders in (\ref{eq:d_approx}) to consider the two cases of $\delta \leq 1$ and $\delta>1$. Note that $d$ in (\ref{eq:d_approx}) satisfies $d>0$ for $\delta<1$ (an exponentially increasing $\tilde{g}(x)$), $d<0$ for $\delta>1$ (an exponentially decaying $\tilde{g}(x)$), and $d=0$ for $\delta=1$ (a nearly uniform distribution). Since $\tilde{g}(0)=c$, we obtain $c$ simply from the exact $g(x)$ in (\ref{eq:g_x_finite_exact}) at $x=0$ after replacing $\pi(K)$ by $\tilde{\pi}(K)$, i.e., $c=\frac{\tilde{\pi}(K)}{\pi(K)}g(0)$. Finally, $\tilde{\pi}(K)$ in (\ref{eq:Pi_K_approx}) guarantees a unit area distribution, i.e., $\int\limits_{0}^{K-n_c M} \tilde{g}(x) \dd x + \sum\limits_{n=0}^{n_c-1}\int\limits_{[K-(n+1)M]^+}^{K-nM}\tilde{g}_n(x) \dd x+\tilde{\pi}(K)=1$.
Next, we study the error associated with the proposed approximation. As far as the performance analysis is concerned, only the pdf in the range $[0,M]$ is needed, c.f. Section \ref{s:BER_outage_analysis}. Assuming $K\!=\!lM$, with $l\!\in\!\mathbb{Z}$ for simplicity, then the approximation error in the range $[0,M]$ is given by\vspace{-0.4cm}
\begin{equation}
\begin{aligned}
e(x)=&g_{l-1}(x)-\tilde{g}(x)=\pi(K)A(x)\Bigg[\sum_{q=0}^{l-1}S_q(x)-\\
&\sum_{q=0}^{l-1}S_q(0)\e^{-W_j(-\delta\e^{-\delta})\frac{x}{M}}\left[\frac{1+\sum_{n=0}^{l-1}C_n}{1+\Sigma_2+\sum_{n=0}^{n_c-1}C_n}\right]\Bigg],
\end{aligned}
\label{eq:error}
\end{equation} where $S_q(x)=\frac{(\delta\e^{-\delta})^q}{q!}\left(q+\frac{\lambda(x-K)}{\delta}\right)^{q-1}\left(q+\frac{\lambda(x-K)}{\delta}-\frac{q}{\delta}\right)$ is the summand in (\ref{eq:g_x_finite_exact}), $A(x)=\lambda\e^{-\lambda(x-K)}$ and $C_n\!=\!\e^{n\delta}\Big(\!\e^{\delta}\!-\!1\!+\!\sum_{q=1}^{n}\!\frac{\left(\delta\e^{-\delta}\right)^q}{q!}\!\left(\e^{\delta}\left(q\!-\!(n\!+\!1)\right)^q\!-\!(q\!-\!n)^q\!\right)\Big)$ is the summand in (\ref{eq:Pi_K_approx}). Using the asymptotic expansion of the exponential of the Lambert W function given by $\frac{\e^{-aW_j(-z)}}{1+W_j(-z)}=\sum_{n=0}^{\infty}(a+n)^n\frac{z^n}{n!}$, it can be shown that if $l\to\infty$, $\sum_{q=0}^{\infty}S_q(x)=\sum_{q=0}^{\infty}S_q(0)\e^{-W_j(-\delta\e^{-\delta})\frac{x}{M}}$ and the approximation error tends to zero. For example, using (\ref{eq:error}) at $n_c=2$, the maximum error percentage $e(x)/g(x)$ in the range of $x=[0,M]$ with $\delta\geq 0.5$ is less than 8.3\% for $K=3M$ and 1.4\% for $K=4M$.

%


\bibliographystyle{IEEEtran}
\bibliography{references}

\begin{thebibliography}{10}
\providecommand{\url}[1]{#1}
\csname url@samestyle\endcsname
\providecommand{\newblock}{\relax}
\providecommand{\bibinfo}[2]{#2}
\providecommand{\BIBentrySTDinterwordspacing}{\spaceskip=0pt\relax}
\providecommand{\BIBentryALTinterwordstretchfactor}{4}
\providecommand{\BIBentryALTinterwordspacing}{\spaceskip=\fontdimen2\font plus
\BIBentryALTinterwordstretchfactor\fontdimen3\font minus
  \fontdimen4\font\relax}
\providecommand{\BIBforeignlanguage}[2]{{%
\expandafter\ifx\csname l@#1\endcsname\relax
\typeout{** WARNING: IEEEtran.bst: No hyphenation pattern has been}%
\typeout{** loaded for the language `#1'. Using the pattern for}%
\typeout{** the default language instead.}%
\else
\language=\csname l@#1\endcsname
\fi
#2}}
\providecommand{\BIBdecl}{\relax}
\BIBdecl

\bibitem{Kansal_2007}
A.~Kansal, J.~Hsu, S.~Zahedi, and M.~B. Srivastava, ``{Power Management in
  Energy Harvesting Sensor Networks},'' \emph{ACM Trans. Embed. Comput. Syst.},
  vol.~6, no.~4, Sep. 2007.

\bibitem{WPC_TDMA}
H.~Ju and R.~Zhang, ``{Throughput Maximization in Wireless Powered
  Communication Networks},'' \emph{IEEE Trans. Wireless Commun.}, vol.~13,
  no.~1, pp. 418--428, January 2014.

\bibitem{WPC_SDMA}
L.~{Liu}, R.~{Zhang}, and K.-C. {Chua}, ``{Multi-Antenna Wireless Powered
  Communication with Energy Beamforming},'' \emph{ArXiv e-prints}, Dec. 2013,
  arXiv:1312.1450.

\bibitem{Sharma2010}
V.~Sharma, U.~Mukherji, V.~Joseph, and S.~Gupta, ``{Optimal Energy Management
  Policies for Energy Harvesting Sensor Nodes},'' \emph{IEEE Trans. Wireless
  Commun.}, vol.~9, no.~4, pp. 1326--1336, April 2010.

\bibitem{Ulukus_Yener_2011}
O.~Ozel, K.~Tutuncuoglu, J.~Yang, S.~Ulukus, and A.~Yener, ``{Transmission with
  Energy Harvesting Nodes in Fading Wireless Channels: Optimal Policies},''
  \emph{IEEE J. Sel. Areas in Commun.}, vol.~29, no.~8, pp. 1732--1743, Sep.
  2011.

\bibitem{WIPT_Architecture_Rui_Zhang_2012}
X.~Zhou, R.~Zhang, and C.~K. Ho, ``{Wireless Information and Power Transfer:
  Architecture Design and Rate-Energy Tradeoff},'' \emph{IEEE Trans. Commun.},
  vol.~61, no.~11, pp. 4754--4767, Nov. 2013.

\bibitem{Culler_2005}
X.~Jiang, J.~Polastre, and D.~Culler, ``{Perpetual Environmentally Powered
  Sensor Networks},'' in \emph{Fourth Intern. Symp. on Information Processing
  in Sensor Networks (IPSN)}, April 2005, pp. 463--468.

\bibitem{Moran_1956}
P.~A.~P. {Moran}, ``{A Probability Theory of a Dam with a Continuous
  Release},'' \emph{The Quarterly Journal of Math.}, vol.~7, pp. 130--137,
  1956.

\bibitem{Infinite_dam_Gani_Prabhu_1957}
J.~Gani and N.~U. Prabhu, ``{Stationary Distributions of the Negative
  Exponential Type for the Infinite Dam},'' \emph{Journal of the Royal
  Statistical Society. Series B (Methodological)}, vol.~19, no.~2, pp.
  342--351, 1957.

\bibitem{Moran_book_1961}
P.~Moran, \emph{{The Theory of Storage}}, ser. Methuen's monographs on applied
  probability and statistics.\hskip 1em plus 0.5em minus 0.4em\relax Methuen,
  1961.

\bibitem{Prabhu_1958}
N.~U. {Prabhu}, ``{On the Integral Equation for the Finite Dam},'' \emph{The
  Quarterly Journal of Mathematics}, vol.~9, pp. 183--188, 1958.

\bibitem{Wang_Giannakis_2003}
Z.~Wang and G.~Giannakis, ``{A Simple and General Parameterization Quantifying
  Performance in Fading Channels},'' \emph{IEEE Trans. Commun.}, vol.~51,
  no.~8, pp. 1389--1398, Aug 2003.

\bibitem{asmussen2003applied}
S.~Asmussen, \emph{Applied Probability and Queues}, ser. Applications of
  mathematics.: Stochastic modelling and applied probability.\hskip 1em plus
  0.5em minus 0.4em\relax Springer, 2003.

\bibitem{Meyn_Tweedie}
S.~Meyn and R.~L. Tweedie, \emph{Markov Chains and Stochastic Stability},
  2nd~ed.\hskip 1em plus 0.5em minus 0.4em\relax New York, NY, USA: Cambridge
  University Press, 1993.

\end{thebibliography}
\end{document}